\documentclass{vldb}[12pt]
\usepackage{amsmath, amstext, amssymb, verbatim}
\usepackage{graphicx}
\usepackage{tikz}
\usetikzlibrary{scopes, arrows, shapes}
\usepackage{tikz-qtree}
\usepackage{subfigure}
\usepackage{pgfplots}
\usepackage{enumitem}
\usepackage{scalefnt}
\usepackage{subfigure}
\newcommand{\gus}[1]{\ensuremath{\mathcal{G}_{(a_{#1},
      \bar{b}_{#1})}}}
\newcommand{\gusi}{\ensuremath{\mathcal{G}_{(1, \bar{1})}}}
\newcommand{\gusz}{\ensuremath{\mathcal{G}_{(0, \bar{0})}}}
\newcommand{\rs}[1]{\ensuremath{\mathcal{R}_{#1}}}
\newcommand{\agus}[1]{\ensuremath{a_{#1}}}
\newcommand{\bgus}[1]{\ensuremath{\bar{b}_{#1}}}
\newcommand{\bt}[1]{\ensuremath{b_{#1}}}
\newcommand{\bgt}[2]{\ensuremath{b_{#1, #2}}}
\newcommand{\SOA}{\stackrel{\mbox{\tiny{SOA}}}{\Longleftrightarrow}}

\newcommand{\lin}[1]{\ensuremath{\mathcal{L}\left(#1\right)}}
\newcommand{\ci}[1]{\ensuremath{\mathcal{T}\left(#1, #1^\prime\right)}}
\newcommand{\se}{\!=\!}

\newtheorem{example}{Example}
\newtheorem{proposition}{Proposition}
\newtheorem{definition}{Definition}

\begin{document}
\title{A Sampling Algebra for Aggregate Estimation}

\numberofauthors{3}

\author{
\alignauthor Supriya Nirkhiwale\\
       %\affaddr{Department of CISE}\\
       \affaddr{University of Florida}\\
       \email{supriyan@ufl.edu}
\alignauthor Alin Dobra\\
       %\affaddr{Department of CISE}\\
       \affaddr{University of Florida}\\
       \email{adobra@cise.ufl.edu}
\alignauthor Christopher Jermaine\\
       %\affaddr{Department of CS}\\
       \affaddr{Rice University}\\
       \email{cmj4@cs.rice.edu}      
}
 
\maketitle

\begin{abstract}
As of 2005, sampling has been incorporated in all major 
database systems. While efficient sampling techniques are realizable,
determining the accuracy of an estimate obtained from the sample is
still an unresolved problem. In this paper, we present a theoretical
framework that allows an elegant treatment of the problem. We base our
work on generalized uniform sampling (GUS), a class of sampling
methods that subsumes a wide variety of sampling techniques. We
introduce a key notion of equivalence that allows GUS sampling
operators to \emph{commute} with selection and join, and derivation of
confidence intervals. We illustrate the theory through extensive
examples and give indications on how to use it to provide meaningful estimations in
database systems.
\end{abstract}

\section{Introduction}

% Sell the idea that sampling is everywhere, can be done in all systems 
% but there is no practical solution for the confidence intervals
% Make thing concrete.

Sampling has long been used by database practitioners to speed up
query evaluation, especially over very large data sets.  For many
years it was common to see SQL code of the form ``\texttt{WHERE} 
\texttt{RAND() > 0.99}''.
Widespread use of this sort of code lead to the inclusion of the
\texttt{TABLESAMPLE} clause in the SQL-2003
standard \cite{sql2003}. Since then, all major databases have incorporated
native support for sampling over relations.
One such query, using the TPC-H
schema, is:
%\begin{example}

\label{ex:query}
\vspace{5 pt}
\noindent
%Query 1: 
\begin{verbatim}
  SELECT SUM(l_discount*(1.0-l_tax))
  FROM lineitem TABLESAMPLE (10 PERCENT), 
       orders TABLESAMPLE (1000 ROWS)
  WHERE l_orderkey  = o_orderkey AND 
        l_extendedprice > 100.0;
\end{verbatim}
%\end{example}
\vspace{5 pt}

% Why confidence intervals/estimates are needed. Unsatisfactory solution otherwise

The result of this query is obtained by taking a Bernoulli sample with
$p=.1$ over \texttt{lineitem} and joining it with a sample obtained 
without replacement (WOR) of size 1000 from \texttt{orders} and
\emph{evaluating} the \texttt{SUM} aggregate. 

In practice, there are two main reasons practitioners write such code.
One is that sampling is useful for debugging expensive queries.
The query can be quickly evaluated over a sample as a sanity check,
before it is unleashed upon the full database.

The second reason is that the  
practitioner is interested in obtaining an idea as to what the actual
answer to the query would be, in less time than would be required to run
the query over the entire database.  This might be useful as a prelude to 
running the query ``for real''---the user might want to see if the result is
potentially interesting---or else the estimate might be used in place of the actual answer.
Often, this situation arises when the
query in question performs an aggregation, since it is fairly intuitive to
most users that sampling can be used to obtain a number
that is a reasonable approximation
of the actual answer.  

The problem we consider in this paper comes from the desire to
use sampling as an approximation methodology.  In this case,
the user is not actually interested in computing an 
aggregate such as ``\texttt{SUM}(\texttt{l\_discount}*(1.0-\texttt{l\_tax)})'' over a sample of the database.
Rather, s/he is interested in estimating the \emph{answer} to such a 
query over the entire database using the sample.  
This presents two obvious problems:
\begin{itemize} %[noitemsep, nolistsep]

\item First, what SQL code should the practitioner write in order to compute
an estimate for a particular aggregate?

\item Second, how does the practitioner have any idea how accurate 
that estimate is?

\end{itemize}

Ideally, a database system would have built-in mechanisms that automatically provide
estimators for user-supplied aggregate queries, and that automatically provide
users with accuracy guarantees.  Along those lines, in this paper
we study how to automatically support SQL of the form:

\vspace{5 pt}
\noindent
%\begin{example}\label{ex:query2}
%Query 2: 
\begin{verbatim}
CREATE VIEW APPROX (lo, hi) AS
SELECT QUANTILE(SUM(l_discount*(1.0-l_tax)), 0.05) 
       QUANTILE(SUM(l_discount*(1.0-l_tax)), 0.95) 
FROM lineitem TABLESAMPLE (10 PERCENT), 
       orders TABLESAMPLE(1000 ROWS)
WHERE l_orderkey  = o_orderkey AND 
        l_extendedprice > 100.0;
\end{verbatim}
%\end{example}

\vspace{5 pt}
Presented with such a query, the database engine will use the user-specified
sampling to
automatically compute two values \texttt{lo} and \texttt{hi} that can be
used as a $[0.05, 0.95]$ confidence bound on the true answer to the
query.  That is, the user has asked the system to compute values
\texttt{lo} and \texttt{hi} such that there is a 5\% chance that the true answer is less than
\texttt{lo}, and there is a 95\% chance that the true answer is less than
\texttt{hi}.  In the general case, the user should be able to specify any 
aggregate over any number of sampled base tables using any sampling scheme, 
and the system would automatically figure
out how to compute an estimate of the desired quantile.  
A database practitioner need have no idea how to compute an estimate for
the answer, nor does s/he need to have any idea how 
to compute confidence bounds; the user only specifies the desired quantiles,
and the system does the rest.

\vspace{5 pt}
\noindent 
\textbf{Existing Work on Database Sampling.}
This is not an easy problem to solve.
While there has been a lot of research on implementing efficient
sampling algorithms \cite{Olken93randomsampling, Lipton1993195}, 
providing confidence intervals for the sample estimate is understood 
only for a few restricted cases. The simplest is when 
only a single relation is sampled.  A slightly more complicated case
was handled by the AQUA system developed at Bell labs \cite{Gibbons98aqua:system, 
Acharya:1999:AAQ, Acharya99joinsynopses, Acharya99aqua:a}.  AQUA considered
\emph{correlated sampling} where a fact table in a star schema is sampled.
These cases are relatively simple because when a single table is sampled,
classical sampling theory applies with a few easy modifications.
Simultaneous work on ripple joins and online aggregation \cite{Haas1996550, Haas_ripplejoins, Haas96large-sampleand, Hellerstein:1997OA, DBLP:conf/sigmod/HaasH01}
extended the class of queries amenable to analysis to include those queries
where multiple tables are sampled with replacement and then joined.

Unfortunately, the extension to other types of sampling is not straightforward,
and to date new formulas have been derived every time a new
sampling is considered (for example, two-table without-replacement sampling
\cite{DBLP:journals/tods/JermaineDAJP06}).
Our goal is to provide a simple theory that makes it possible
to handle very general types of queries over virtually any uniform sampling
scheme: with replacement sampling, fixed-size without replacement sampling,
Bernoulli sampling, or whatever other sampling scheme is used. 
The ability to easily handle arbitrary types of sampling is especially
important given that the current SQL standard allows for a somewhat
mysterious \texttt{SYSTEM} sampling specification, whose exact implementation
(and hence its statistical properties) are left up to the database designers.
Ideally, it should be easy for a database designer to apply our theory to an
arbitrary \texttt{SYSTEM} sampling implementation.

\vspace{5 pt}
\noindent 
\textbf{Generalized Uniform Sampling.}
One major reason that new theory and derivations were previously required for each 
new type of sampling is that the usual analysis is tuple-based, where
the inclusion probability of each tuple in the output set is used as the
basic building block; computing expected values and variances requires
intricate algebraic manipulations of complicated summations.
In previous work, we defined a notion that we called
\emph{Generalized Uniform Sampling} (GUS) \cite{turbodbo} that
subsumes many different sampling schemes (including
all of the aforementioned ones, as well as block-based variants thereof).
In this paper, we
develop an algebra over many common relational operators, as well
as the GUS operator.
This makes it possible to take any query plan that contains one or more
GUS 
operators and the supported relational operators, and perform a statistical
analysis of the accuracy of the result in an algebraic fashion, working from
the leaves up to the top of the plan.

No complicated algebraic manipulations over nested summations are required.
This algebra can form the basis for a lightweight tool for providing estimates
and quantiles,
that should be
easily integrable into any database system.
The database need only feed the tool the user-specified quantiles,
the set of tuples returned
by the query, some simple lineage information over those result
tuples, and the query plan, and the 
tool can automatically compute the desired quantiles.

\vspace{5 pt}
\noindent 
\textbf{Our Contributions.}
The specific contributions we make in this paper are:
\begin{itemize} %[noitemsep, nolistsep]
\item 
We define the notion of Second Order Analytical (SOA)-equivalence, a key equivalence relationship
 between query plans that is strong enough to allow quantile analysis but 
weak enough to ensure commutativity of sampling and relational operators.
 
\item
We define the GUS operator that emulates a wide class of sampling methods. 
This operator commutes with most relational operators under  SOA-equivalence.
 
\item We develop an algebra over GUS and relational operators that allows 
derivation of SOA-equivalent plans. These plans easily allow moment calculations that can be used
to estimate quantiles.

\item 
We describe how our theory can be used to add estimation capabilities
to existing databases so that the required changes to the query optimizer and execution
engine are minimal. Alternatively, the estimator can be implemented as
an external tool.
\end{itemize} 

Our work provides a straightforward analysis for the SUM aggregate. It can be easily extended for  COUNT by substituting the aggregated attribute to 1 and applying the analysis for SUM on this attribute. Though the analysis for AVERAGE presents a slightly non-linear case, the analyses for SUM and COUNT lay a foundation for it. The confidence intervals can be derived using a method for approximating probability distribution/variance such as the delta method. The analysis for MIN, MAX and DISTINCT are extremely hard problems to solve due to their non-linearity. For example DISTINCT requires an estimate of all the distinct values in the data and the number of such values. It is thus beyond the scope of this paper.
 
While selections and joins are the highlight of our paper, we show that SOA-equivalence allows analysis for other database operators like cross-product (compaction), intersection (concatenation) and union.  Due to space constraints, we are unable to include all technical proofs, implementation details and discussions. These are available in the the extended version of this paper\cite{}.

% Structure of the rest of the paper
The rest of the paper is organized as follows. In Section~\ref{sec:relwork}, we provide a brief 
overview of related work in this area. In Section~\ref{sec:gus}, we introduce GUS methods
and give details on how to get estimates and confidence intervals for them. In 
Section~\ref{sec:analysis}, we introduce the notion of SOA-equivalence between query plans and 
prove that GUS operators commute with a variety of relational operators in the SOA sense. In Section~\ref{sec:properties}, we investigate interactions between GUS 
operators when applied to the same data and explore more possibilities
in using them. In Section~\ref{sec:implementation}, we provide insight on how our theory can be used to implement a separate add-on tool and how to enhance the performance of the variance estimation. In Section~\ref{sec:exp}, we test our implementation thoroughly, and provide accuracy and runtime analysis. We explore some possible applications in Section~\ref{sec:applic} and conclude with a discussion in Section~\ref{sec:conclusion}.

\section{Related Work}\label{sec:relwork}
% Starting papers
The idea of using sampling in databases for deriving estimates for a single relation was first studied by 
Shapiro et al. \cite{Piatetsky-Shapiro:1984}. Since then, much research has focused 
on implementing efficient sampling algorithms in databases \cite{Olken93randomsampling, Lipton1993195}.
Providing confidence intervals on estimates for SQL aggregate queries is difficult, which is why 
there has been limited progress in this area. Olken \cite{Olken93randomsampling} studied the problem for 
specific sampling methods for a single relation. This line of work ended abruptly when Chaudhuri 
et al. \cite{DBLP:journals/debu/ChaudhuriM99, DBLP:conf/sigmod/ChaudhuriMN99} proved that extracting IID samples from a join of two relations is infeasible.

% Is there a solution in practice. DB2 does AQUA sampling
% not enough.

Another line of research was the extension to the
\emph{correlated sampling} pioneered by the AQUA system \cite{Gibbons98aqua:system, Acharya99aqua:a, Acharya:1999:AAQ}.
AQUA is applicable to a 
star schema, where the goal is sampling from the fact
  table, and including all tuples in dimension tables that
match selected fact table tuples. 
The AQUA type of sampling has been
incorporated in DB2 \cite{Gryz:2004:QSD:1007568.1007664}.

% Why the AQUA sampling works and why is it not generalizable?

The reason confidence intervals can be provided for AQUA type sampling
is the fact that \emph{independent identically distributed} (IID)
samples are obtained from the set over which the aggregate is computed. 
A straightforward use of the \emph{central limit theorem}
readily allows computation of good estimates and confidence
intervals. Indeed, it is widely believed \cite{DBLP:journals/debu/ChaudhuriM99, 
DBLP:conf/sigmod/ChaudhuriMN99, Olken93randomsampling, 
Gibbons98aqua:system, Acharya99aqua:a, Acharya:1999:AAQ}
that IID samples at the top of the query plan are required to
provide \emph{any} confidence interval. This idea leads to the search
for a \emph{sampling operator} that commutes with database
operators. This endeavor proved to be very difficult from the
beginning \cite{DBLP:journals/debu/ChaudhuriM99} when joins are involved. To see why this is the
case, consider a tuple $t\in$ \texttt{orders} and two tuples $u_1,u_2$
in \texttt{lineitem} that join with $t$ (i.e. they have the same value
for \texttt{orderkey}). Random selection of tuples $t,u_1,u_2$ in the
sample does not guarantee random selection of result tuples $(t,u_1)$
and $(t,u_2)$.  If $t$ is not selected, neither tuple can exist, and thus
sampling is correlated. A lot of effort \cite{DBLP:journals/debu/ChaudhuriM99, 
DBLP:conf/sigmod/ChaudhuriMN99} has been spent in finding practical 
ways to de-correlate the result tuples with only limited success. 

% The analytical solution for the whole query (Haas, DBO, TurboDBO)
% why does it avoid the earlier solutions.

Progress has been made using a different line of thought by
Hellerstein and Hass \cite{Hellerstein:1997OA} and the generalization in
\cite{Haas:1996:SCE:233502.233514} for the special case of sampling with
  replacement. The problem of producing IID result samples is avoided
by developing central limit theorem-like results for the combination
of relation level sampling with replacement. 
The theory was generalized first to \emph{sampling without 
replacement} for single join queries \cite{DBLP:journals/tods/JermaineDAJP06}, then further
generalized to arbitrary uniform sampling over base directions and
arbitrary \texttt{SELECT-FROM-WHERE} queries without duplicate
elimination in DBO \cite{dbo}, and finally to allow sampling
across multiple relations in Turbo-DBO \cite{turbodbo}. Even though
some simplification occurred through these theoretical developments,
they are mathematically heavy and hard to understand/interpret. 
Moreover, the theory, especially DBO and Turbo-DBO, is tightly coupled 
with the systems developed to exploit it. 

% Problem with analytical solution. Why it did not lead to practical 
% solutions
% end by suggesting what key idea can be reused: GUS
% core technical hurdle solved.

Technically, one major problem in all the mathematics used to analyze 
sampling schemes
is the fact the analyses use functions and summations over tuple
domains, and not the operators and algebras that the database community is
used to. This makes the theory hard to comprehend and apply. The fact
that no database system picked up these ideas to provide a confidence
interval facility is a direct testament of these difficulties. 

%While the theory is complicated and hard to deal with, the main

%technical difficultly is overcome in Turbo-DBO in a very general sense

%through the notion of \emph{Generalized Uniform Sampling}.

\section{Generalized Uniform Sampling}\label{sec:gus}

Previous attempts at accommodating a sampling operator in a query
plan were limited to specific sampling methods. 
In previous work
\cite{turbodbo}, we analyzed a large class of sampling methods for
which the analysis can be unified: Generalized Uniform
Sampling(GUS). Sampling methods such as uniform sampling with/without
replacement, Bernoulli sampling and more elaborate strategies like the
chaining in \cite{turbodbo} are members of the GUS family. Moreover, the
variance of any GUS sampling can be efficiently estimated. We briefly introduce GUS sampling methods in this section and 
investigate them further in this paper. 

\begin{definition}[GUS Sampling \cite{turbodbo}]\label{def:gus}
A randomized selection process $\gus{}$ which gives a sample $\mathcal{R}$ from 
${\bf R} = R_1 \times R_2 \times \cdots \times R_n$ 
is called Generalized Uniform Sampling (GUS) method,
, if, for any given tuples
$t = (t_1,\dots,t_n)\in \bf{R}$
$t^\prime = (t_1^\prime,\dots,t_n^\prime)$, $P(t \in \rs{})$ is independent of $t$, 
and $P(t,t^\prime \in \rs{})$ depends only on $\{i : t_i = t_i^\prime\}$. In such a case, the 
GUS parameters  \agus{}, $\bgus{} = \{\bt{T}|T \subset \{1:n\}\}$ are defined as:
\begin{align*}
\agus{} &= P[t \in \rs{}] \\
\bt{T} &= P[t \in \rs{} \land t^\prime \in \rs{}|
\forall i\in T, t_i = t_i^\prime, \forall j\in T^C, t_j\neq t_j^\prime]. 
\end{align*}
\end{definition}

This definition requires GUS sampling to behave like a
\emph{randomized filter}. In particular, any GUS operator can be
viewed as a selection process from the underlying data, a process that
can introduce correlations. The uniformity of GUS requires that the
randomized filtering is performed on \emph{lineage} of tuples  and not on
the content. 
% More specifically, two tuples $t = (t_1,\dots,t_n), t^\prime = (t_1^\prime, 
% \dots, t_n^\prime)$ satisfy $t = t^\prime$ if $(t_1 = t_1^\prime, \dots, t _n = t_n^\prime)$.
% The equality condition depends on what base relations were involved in the tulpes. 
As simple as the idea is, expressing any sampling process in the form of GUS is a non-trivial 
task. Example~\ref{ex:GUS_join_param} shows the calculation of GUS parameters for a simple case.

\begin{example}\label{ex:GUS_join_param}
In this example we show how the GUS definition above  can be used to characterize the estimation 
necessary for the query from the paper's introduction. We denote by \texttt{l\_s} the Bernoulli sample with $p=0.1$ 
from 
\texttt{lineitem} and by \texttt{o\_s} the WOR sample of size $1000$ from \texttt{orders}. We assume that 
cardinality of \texttt{orders} is $150000$. Henceforth, for ease of exposition, we will denote all base relations 
involved by their first letters. For example, \texttt{lineitem} will be denoted by \texttt{l}. 

Applying the definition above and the independence between sampling processes, we can derive 
the parameters for this GUS as follows:
For any tuple $t \in \texttt{lineitem}$ and tuple $u \in \texttt{orders}$: 
$$
a = P[(t \in \texttt{l\_s}) \land (u \in \texttt{u\_s})] = 0.1 \times \frac{1000}{150000} = 6.667 
\times 10^{-4} 
$$

\noindent
since the base relations are sampled independently from each other. For any tuples $t, t^\prime \in 
\texttt{lineitem}$ and $u, u^\prime \in \texttt{orders}$:
\noindent 
\begin{eqnarray*}
b_{\varnothing}
&=& P[(t, t^\prime \in \texttt{l\_s}) \land (u, u^\prime) \in \texttt{o\_s}]\\ 
&=& 0.1 \times 0.1 \times \frac{1000}{150000} \times \frac{999}{149999}\\ 
&=& 4.44 \times 10^{-7}, 
\end{eqnarray*}
\noindent
and 
\begin{eqnarray*}
b_{o} 
&=& P[t \in \texttt{l\_s}] \times P[t^\prime \in \texttt{l\_s}|t \in \texttt{l\_s}] \times P[u \in 
\texttt{o\_s}] \\
&=& 0.1 \times 0.1 \times \frac{1000}{150000} = 6.667 \times 10^{-5}. 
\end{eqnarray*}
\noindent
Similarly,
\begin{eqnarray*}
b_{l}
&=& P[(t \in \texttt{l\_s}) \land (u, u^\prime \in \texttt{o\_s})]\\
&=& P[t \in \texttt{l\_s}] \times P[u \in \texttt{o\_s}]  \times P[u^\prime \in \texttt{o\_s}|u \in 
\texttt{o\_s}]\\ 
&=& 0.1 \times \frac{1000}{150000} \times \frac{999}{149999}\\
&=& 4.44 \times 10^{-6}. 
\end{eqnarray*} 
 
\noindent
The last term is 
$$
b_{l,o} = P[ (t \in \texttt{l\_s}) \land (u \in \texttt{o\_s})] = 0.1 \times \frac{1000}{150000} = 6.667 
\times 10^{-4}. 
$$
\noindent
Notice that the GUS \emph{captures} the entire estimation process, not
only the two individual sampling methods. The above analysis dealt with a simple join consisting of 
two base relations. For more complex query plans, the derivation of GUS parameters would involve 
consideration of all possible interactions between participating tuples. This will make the analysis 
highly complex. 
\end{example}

The analysis of any GUS sampling method
for a SUM-like aggregate is given as follows. 

\newtheorem{thm}{Theorem}
\begin{thm}\label{thm:ev}\cite{turbodbo}
  Let $f(t)$ be a function/property of $t\in R$, and \rs{} be the sample obtained by a GUS method 
  \gus{}. Then, the aggregate $\mathcal{A}=\sum_{{\bf t} \in R} f({\bf t})$ and the sampling estimate 
  $ X = \frac{1}{a} \sum_{{\bf t} \in \mathcal{R}} f({\bf t})$ have the property: 
  \begin{align}
    E[X] &= \mathcal{A} \nonumber\\
    \sigma^2(X) &= \sum_{S \subset \{1:n\}}\frac{c_S}{a^2}y_S - y_{\phi} \label{gusvariance}
  \end{align}
  with 
  \begin{align*}
    y_S &= \sum_{t_i \in R_i| i \in S} \left(\sum_{t_j \in R_j| j \in S^C}f({t_i, t_j})\right)^2\\
    c_S &= \sum_{T \in \mathcal{P}(n)} \left(-1 \right)^{|T| + |S|} b_T. 
  \end{align*}
\end{thm}

The above theorem indicates that the GUS estimates of SUM-like aggregates are unbiased and 
that the variance is simply a linear combination of properties of the data, terms $y_S$ 
and properties of the GUS sampling method $c_S$. Moreover, $y_S$ can be estimated from 
samples of \emph{any} GUS \cite{turbodbo}. This result is not asymptotic; it gives the 
exact analysis even for very small samples. Once the estimate and the variance 
are computed, confidence intervals can be readily provided using either
the normality assumption or the more conservative \emph{Chebychev} bound \cite{turbodbo}.

% \begin{example}
% Using the GUS parameters of Query 1 computed in Example~\ref{ex:GUS_q}, we immediately get:
% $c_S = $. From the samples obtained, we get $y_\phi = $ and $y_S = $. Thus $E[X] = $ and 
% $\sigma^2(X) = $. This variance and the expected value can be readily used to compute confidence 
% intervals. If we use the normal approximation, we get.... 
% \end{example}

In the rest of the paper, we will study GUS sampling methods in detail.
%with various 
%parameters applied to various input spaces. 

\section{Analysis of Sampling Query Plans}\label{sec:analysis}

The high-level goal of this paper, is to introduce a tool that
computes the confidence bounds of estimates based on sampling. Given a
query plan with sampling operators interspersed at various points, our
tool transforms it to an\\\emph{analytically equivalent} query plan that
has a particular structure: all relational operators except the final
aggregate form a subtree that is the input to a single GUS sampling
operator. The GUS operator feeds the aggregate operator that produces
the final result. Note that this transformation is done solely for the
purpose of computing the confidence bounds of the result; it does not
provide a better alternative to the execution plan used as input.
Once this transformation is accomplished, Theorem~\ref{thm:ev}
readily gives the desired analysis -- the equivalence ensures that the
analysis for the special plan coincides with the analysis for the
original plan.

A natural strategy to obtain the desired structure is to 
perform multiple local transformations on the original query
plan. These local transformations are based on a notion of analytical
equivalence, that we call Second Order Analytical (SOA) equivalence.  They 
allow both commutativity of relational and GUS operators, and
consolidation of GUS operators. Effectively, these local transformations
allow a plan to be put in the special form in which there is a single
GUS operator just before the aggregate. 

In this section, we first define the
SOA-equivalence and then use it to provide \emph{equivalence} relationships
that allow the plan transformations mentioned. A more elaborate
example showcases the theory in the latter part of the section.

\subsection{SOA-Equivalence}\label{subsec:SOA}

The main reason the previous attempts to design a \emph{sampling
 operator} were not fully successful is the requirement to ensure IID
samples at the top of the plan. Having IID samples makes the analysis 
easy since Central Limit Theorem readily provides confidence intervals. 
However it is too restrictive to allow plans with multiple joins to be dealt
with. It is important to notice that the difficulty is not in 
\textit{executing} query plans containing sampling but in 
\textit{analyzing} such query plans.

The fundamental question we ask in this section is: What is the least
restrictive requirement we can have and still produce useful
estimates?  Our main interest is in how the requirement can
be transformed into a notion of \emph{equivalence}. This will enable us 
to talk about \emph{equivalent plans}, initially, but more usefully
about \emph{equivalent expressions}. The key insight comes from the
observation that it is enough to compute the expected value and
variance of any approximate query plan. Then either the
conservative Chebychev bounds or the optimistic\footnote{While the CLT
  theorem does not apply due to the lack of IID samples, the
  distribution of most complex random variables made out of many
  loosely interacting parts tends to be normal.} normal-distribution
based bounds can be used to produce confidence intervals. Note that 
confidence intervals are the end goal, and, preserving expected value 
and variance is enough to guarantee the same confidence interval using 
both CLT and Chebychev methods.

Thus, for our purposes, \emph{two query plans are equivalent if their result has the same
  expected value and variance}. This equivalence relation between
plans already allows significant progress. It is an extension of the
classic plan equivalence based on obtaining the same answer to
\emph{approximate/randomized} plans. 
From an operational sense, though,
the plan equivalence is not sufficient to provide interesting
characterizations. The main problem is the fact that the equivalence
exists only between complete plans that compute aggregates. It is not
clear what can be said about intermediate results---the equivalent of
non-aggregate relational algebra expressions. 

The key to extend the equivalence of plans to equivalence of
expressions is to first design such an extension for the classic
relational algebra. To this end, assume that we can only use equality
on numbers that are results of SUM-like aggregates but we cannot
directly compare sets. To ensure that two expressions are equivalent,
we could require that they produce the same answer using
\textbf{any} SUM-aggregate. Indeed, if the expressions produce the
same relation/set, they \emph{must} agree on \emph{any} aggregate
computation using these sets since aggregates are deterministic and,
more importantly, do not depend on the order in which the computation is
performed. The SUM-aggregates are crucial for this definition since
they form a vector space. Aggregates $A_t$ that sums function
$f_t(u)=\delta_{tu}$ are the basis of this vector space; agreement on
these aggregates ensures set agreement. Extending these ideas to
randomized estimation, we obtain the following. 
\begin{definition}[SOA-equivalence]
Given (possibly randomized) expressions $\mathcal{E}(R)$ and $\mathcal{F}(R)$, we say
\begin{equation*}
\mathcal{E}(R)\SOA\mathcal{F}(R)
\end{equation*}
if for \textbf{any arbitrary} SUM-aggregate $\mathcal{A}_f(S)=\sum_{t\in S} f(t)$
\begin{align*}
E[ \mathcal{A}_f( \mathcal{E}(R) ) ]  &= E[ \mathcal{A}_f( \mathcal{F}(R) ) ]  \\
Var[ \mathcal{A}_f( \mathcal{E}(R) ) ]  &= Var[ \mathcal{A}_f( \mathcal{F}(R) ) ].  \\
\end{align*}
\end{definition}

From the above discussion, it immediately follows that
SOA-equivalence is a generalization and implies set equivalence for
non-randomized  expressions, as stated in the following proposition. 
\begin{proposition} Given two relational algebra expressions $E(R)$
  and $F(R)$ we have:
  \begin{equation*}
    E(R)= F(R) \Leftrightarrow E(R) \SOA F(R)
  \end{equation*}
\end{proposition}

The next proposition establishes that SOA-equivalence is indeed an \emph{equivalence relation} 
and can be manipulated like relational equivalence. 
\begin{proposition}
SOA-equivalence is an equivalence relation, i.e., for any expressions $\mathcal{E},\mathcal{F},\mathcal{H}$ and relation
$R$:
\begin{gather*}
  \mathcal{E}(R) \SOA \mathcal{E}(R) \\
  \mathcal{E}(R) \SOA \mathcal{F}(R) \Rightarrow \mathcal{F}(R) \SOA
  \mathcal{E}(R) \\
  \mathcal{E}(R) \SOA \mathcal{F}(R) \land \mathcal{F}(R) \SOA
  \mathcal{H}(R)
  \Rightarrow \mathcal{E}(R) \SOA \mathcal{H}(R). 
\end{gather*}
\end{proposition}

SOA-equivalence subsumes relational algebra equivalence. The strength of
 SOA-equivalence is the fact that it does not depend on a notion of randomized set
equivalence, an equivalence that would be hard to define especially if it
has to preserve aggregates. 

\begin{proposition}\label{thm:SOA-P} Given two relational algebra expressions $\mathcal{E}(R)$
  and $\mathcal{F}(R)$ we have:
  \begin{gather*}
    \mathcal{E}(R) \SOA \mathcal{F}(R)\\
    \Leftrightarrow\\
    \forall t \in R, \ \ P[t \in \mathcal{E}(R)] = P[t \in \mathcal{F}(R)] \; \mbox{ and }\\
    \forall t,u \in R, \ \  P[t,u \in \mathcal{E}(R)] = P[t,u \in \mathcal{F}(R)] \; 
  \end{gather*}
\end{proposition}

Proposition~\ref{thm:SOA-P} provides a powerful alternative to
SOA-equivalence. This equivalence is in terms of first and second order probabilities, and we refer to it as \emph{SOA-set equivalence}. Another way to interpret the result above is that SOA-set equivalence
is the same as agreement on all SUM-like aggregates. More importantly for this
paper, SOA-set equivalence provides an alternative proof technique to
show SOA-equivalence. Often, proofs based on SOA-set equivalence are
simpler and more compact.

Section~\ref{sec:implementation} contains a recipe for expected value and variance
computation for a specific situation, when there is a single overall
GUS sampling on top.  Starting with the given query plan that contains
both sampling and relational operators, if we find a SOA-equivalent
plan that is equivalent and has no sampling except a GUS at the top, we
readily have a way to compute the expected value and variance of the
original plan. In the rest of this section we pursue this idea further and
show how SOA-equivalent plans with the desired structure can be obtained from a general 
query plan.

\subsection{GUS Quasi-Operators}\label{sec:gus-op}

Except under restrictive circumstances, the sampling operators will
not commute with relational operators. This, as we mentioned is the
main reason previous work made limited progress on the issue. As we
will see later in this section, GUS sampling 
does commute in a SOA-equivalence
sense with most relational operators. The reason we can commute GUS (but not
specific sampling methods) is that, due to its generality, it
can \emph{capture} the correlations induced by the relational
operators. The first step in our analysis has to be a
\emph{translation} from specific sampling to GUS-sampling. 

\begin{figure} \centering
\label{tbl:bwor}
\scalebox{1}{
\begin{tabular}{c|c}
	Sampling method & GUS parameters \\ \hline
	Bernoulli($p$) & $\agus{} = p$, $\bt{\varnothing} = p^2$, \bt{\texttt{R}} $= p$\\
	WOR ($n$, $N$) & $\agus{} = \frac{n}{N}, \bt{\varnothing} = \frac{n(n-1)}{N(N-1)}, \bt{\texttt{R}} = \frac{n}{N}$
\end{tabular}}
\caption{{\small GUS parameters for known sampling methods on a single relation}}
\end{figure}

Before we talk about the translation from sampling to GUS operators, we
need to clarify and refine the Definition~\ref{def:gus} of GUS
sampling. As part of the definition, terms of the form
$t_i=t_i^\prime$ or $t_j\neq t_j^\prime$ are used. The meaning of
these terms is somewhat fuzzy in both \cite{dbo} and
\cite{turbodbo}. Intuitively, they capture the idea that tuples (or
parts) are the same or different. Since in this paper we will have
multiple GUS operators involved, it is important to make the meaning
of such terms very clear. We do this through a notion that proved
useful in probabilistic databases (among other uses):
\emph{lineage}\cite{provenanceindb}. Lineage allows dissociation of  the ID of
a tuple from the content of the tuple, for base relation, and tracking
the composition of derived tuples. With this, $t_i=t_i^\prime$ means
that the two tuples are the same -- have the same ID/lineage -- not that they
have the same content. 

Representing and manipulating lineage is a complex subject. In this
work, since we only accommodate selection and joins the issue is
significantly simpler. The selection leaves lineage unchanged, the
lineage of the result of the join is the union of the lineage of the
matching tuples. Thus, lineage can be represented in relational form
with one attribute for each base relation participating in the
expression. We can thus talk about \emph{lineage schema}
$\mathcal{L}(R)$, a synonym of the set of base relations participating
in the expression of $R$. The lineage of a specific tuple $t\in R$ will
have values for the lineage of all base relations constituting $R$.
A particularly useful notation related to lineage is:
$\ci{t}=\{R_k|t_k = {t_k}^\prime, k\in\lin{R}\}$, the common part of
the lineage of tuples $t$ and $t^\prime$, i.e. the base relations on
which the lineage of $t$ and $t^\prime$ agree.

\begin{figure}[h]\label{fig:Query1}
\centering
\scalebox{0.7}{
\begin{tabular}{@{\extracolsep{0em}}ccc}
\Tree [.\texttt{SUM} [.$\Join$ [.$B_{0.1}$ [.\texttt{l} ] ] [.$WOR_{1000}$ [. \texttt{o} ] ] ] ] & 
\Tree [.\texttt{SUM} [.$\Join$ [.\gus{B} [.\texttt{l} ] ] [.\gus{W} [.\texttt{o} ] ] ] ] & 
\Tree [.\texttt{SUM} [.\gus{BW} [.$\Join$ [.\texttt{l} ] [.\texttt{o} ] ] ] ]\\
(a) & (b) & (c)
\end{tabular}}
\caption{{\small Query 1}}
\end{figure}

\begin{example}\label{ex:baseparam}
 The query from the introduction uses two sampling methods:
 Bernoulli sampling with $p = 0.1$ on \texttt{lineitem} and sampling 1000 tuples without replacement
 from \texttt{orders}(150,000 tuples). These methods can be expressed in terms of GUS as \gus{B} and \gus{W} as follows:
For \gus{B}:
$\agus{B} = 0.1$ and $\bgus{B} = \{\bgt{B}{\varnothing}, \bgt{B}{\texttt{l}}; \bgt{B}{\varnothing} = 0.01, \bgt{B}{\texttt{l}} = 0.1\}$
For \gus{W}:
$\agus{W} = 6.667 \times 10^{-3}$ and $\bgus{W} = \{\bgt{W}{\varnothing}, \bgt{W}{\texttt{o}}; \bgt{W}{\varnothing} = 4.44 \times 10^{-5}, \bgt{W}{\texttt{o}} = 6.667 \times 10^{-3}\}$
\end{example}

It is important to note that the GUS is not an operator but a
quasi-operator. While it corresponds to a real operator when the
translation from specific sampling to GUS happens, it will not
correspond to an operator after transformations. There is no
  need to provide or even to consider  an implementation of a general
  GUS operator since GUS will only be used for the purpose of
  analysis. 

In the rest of this section, we will assume that all specific sampling
operators were replaced by GUS quasi-operators, thus will not be
encountered by the re-writing algorithm.
We designate by $\gus{}(R)$, a GUS method applied 
to a relation $R$, and the resulting sample by \rs{}. When multiple GUS methods are used, 
the $i$'th GUS method and its resulting sample will be denoted by \gus{i} and \rs{i} respectively. 
  
\begin{figure}\label{table:notation}
\centering
\scalebox{1}{
\begin{tabular}{ l | p{.7\columnwidth}}
	Notation & Meaning  \\ \hline
	\rs{} & Random subset of $R$\\
	\agus{} & $P[t \in \rs{}]$\\
	\lin{R} & Lineage schema of $R$\\
	\lin{t} & Lineage of tuple $t$\\
	$T$ & Subset of $\lin{R}$\\
	\ci{t} & $\{R_k|t_k = {t_k}^\prime, k\in\lin{R}\}$\\
	\bt{T} & $P[t, t^\prime \in \rs{}| T = \ci{t}]$\\
	\bgus{} & $\{\bt{T}| T \in \mathcal{P}(n)\}$\\ 
	\gus{} & GUS method with parameters \agus{} and \bgus{}\\
	$\gus{}(R)$ & \gus{} applied to relation $R$ 
\end{tabular}}
\caption{{\small Notation used in paper}}
\end{figure}

\subsection{Interaction Between GUS and Rel Ops}\label{sec:relGUS}
% We need 1.union(segregated aggregates from diff lineages) 2. join 3.selection 

As we stated in Section~\ref{subsec:SOA}, SOA-equivalence is the key
for deriving an \emph{analyzable} plan that is \emph{equivalent} to the one
provided by the user. The results in this section provide equivalences
that allow such transformations that lead to a single, top, GUS
operator. The results in this section make use of the notation in
Table~\ref{table:notation}.

\begin{proposition}[Identity GUS] \label{thm:identity} The
  quasi-operator $\gusi$, i.e. a GUS operator with
  $a=1$, $b_T=1$, can be inserted at any point in a query plan without
  changing the result. 
\end{proposition} 
\begin{proof}
  Since $\agus{} = 1$, all input tuples are allowed with probability 1, i.e., no filtering happens. 
\end{proof}

\begin{proposition} [Selection-GUS Commutativity] \label{thm:select}
For any $R$, selection $\sigma_C$ and GUS $\gus{}$, 
\begin{equation*}
\sigma_C (\gus{}) (R) \stackrel{\mbox{\tiny{SOA}}}{\Longleftrightarrow} \gus{}(\sigma_C(R)). 
\end{equation*}
\end{proposition}
\begin{proof}
Let $R^\prime = \sigma_C(R)$. On computing $\rs{} \cap R^\prime$ we see that 
\begin{gather*}
\forall(t \in R^\prime),  P[t \in \rs{} \cap R^\prime]  = P[t \in \rs{}]I_{\{t \in R^\prime\}} = a.
\end{gather*}
\begin{gather*}
\forall(t,t^\prime \in R^\prime), P[t, t^\prime \in \rs{} \cap R^\prime| T \se \ci{t}] \\ = P[t, t^\prime \in \rs{}| T \se \ci{t}] = \bt{T} \qed 
\end{gather*}
\end{proof}

The above results are somewhat expected and have been covered for
particular cases in previous literature. The following result, though,
overcomes the difficulties in \cite{DBLP:conf/sigmod/ChaudhuriMN99}.

\begin{proposition}[Join-GUS Commutativity] \label{thm:join}
For any $R,S$, join $\Join_{\theta}$ and GUS methods $\gus{1},\gus{2}$, 
if $\lin{R_1}\cap \lin{R_2}=\varnothing$
\begin{gather*}
\gus{1}(R_1) \Join_{\theta} \gus{2}(R_2) \SOA \gus{}(R_1 \Join_{\theta} R_2),\\
where,\qquad a = \agus{1}\agus{2},\qquad \bt{T} = \bgt{1}{T_1}\bgt{2}{T_2}
\end{gather*}
\noindent
with $T_1 = T \cap \lin{R_1}$ and $T_2 = T \cap \lin{R_2}$.
\end{proposition}
\begin{proof}
  We proved in Proposition~\ref{thm:select} that a GUS method commutes
   with selection. Thus, it is enough to prove commutativity of a GUS method
   with cross product.
 Let $R = R_1 \times R_2$ and $t = (t_1, t_2), t^\prime = (t_1^\prime, t_2^\prime) \in
  R$. Thus, $\lin{R} = \lin{R_1} \cup \lin{R_2}$. We have:
  \begin{align*}
    a & = P[t\in\rs{}] = P[t_1 \in \rs{1} \land t_2 \in \rs{2}]\\
    & = P[t_1 \in \rs{1} \land t_2 \in \rs{2}] = \agus{1}\agus{2}.
  \end{align*}
  Since $\lin{R_1} \cap \lin{R_2} = \varnothing$,  for an arbitrary 
$T \in \lin{R}$, $T_1 = T \cap \lin{R_1}$ and $T_2 = T \cap \lin{R_2}$ 
we have, $T_1 \cap T_2 = \varnothing$ (disjunct lineage). With this, we first get:
\begin{equation*}
\ci{t} \se T \Leftrightarrow \ci{t_1} \se T_1 \land \ci{t_2} \se T_2
\end{equation*}
and then using the above and independence of GUS methods, 
  \begin{align*}
   \bt{T} & = P[t \in \rs{} \land t^\prime \in \rs{}|\ci{t} \se T]\\
	& = P[t_1, t_1^\prime \in \rs{1} \land t_2, t_2^\prime \in \rs{2}|\ci{t_1} \se T_1 \land \ci{t_2} \se T_2]\\
	& = P[t_1, t_1^\prime \in \rs{1}|\ci{t_1} \se T_1] P[t_2, t_2^\prime \in \rs{2}|\ci{t_2} \se T_2]\\
	& = \bgt{1}{T_1}\bgt{2}{T_2}. \qed
  \end{align*}
\end{proof}

\begin{example} \label{ex:join}
Applying the above results to the GUS co-efficients obtained in Example \ref{ex:baseparam}, 
we can derive the following co-efficients for \gus{} in Fig \ref{fig:Query1}:
\begin{gather*}
 \agus{} = \agus{1}\agus{2}= 0.1 \times 6.667 \times 10^{-3} = 6.667 \times 10^{-4}.\\
% \bgus{} = \{\bt{\varnothing}, \bt{\texttt{l}}, \bt{\texttt{o}}, \bt{\texttt{lo}}\}, where,\\
\bt{\varnothing} = \bgt{1}{\varnothing}\bgt{2}{\varnothing} = 0.01 \times 4.44 \times 10^{-5} = 4.44 \times 10^{-7}.\\
\bt{\texttt{o}} = \bgt{1}{\varnothing}\bgt{2}{\texttt{o}} = 0.01 \times 6.667 \times 10^{-3} = 6.667 \times 10^{-5}.\\
\bt{\texttt{l}} = \bgt{1}{\texttt{l}}\bgt{2}{\varnothing} =0.1 \times 4.44 \times 10^{-5} = 4.44 \times 10^{-6}.\\
\bt{\texttt{lo}} = \bgt{1}{\texttt{l}}\bgt{2}{\texttt{o}} = 0.1 \times 6.667 \times 10^{-3} = 6.667 \times 10^{-4}.
\end{gather*}
\end{example}

\begin{example}\label{ex:large}
  In this example we provide a complete walk-through for a larger
  query plan. The input is the query plan in Figure~\ref{fig:large}.a
  that contains 3 sampling operators, 3 joins and refers to relations
  \texttt{lineitem, orders, customers} and \texttt{part}. To analyze
  such a query, the first step is to re-write the sampling operators
  as GUS quasi-operators \gus{1},\\\gus{2}, \gus{3} as in
  Figure~\ref{fig:large}.b. The second step, shown in
  Figure~\ref{fig:large}.c is to apply Proposition~\ref{thm:join} to
  commute \gus{1} and \gus{2} with the join resulting in \gus{12}
  . This step also shows the application of
  Proposition~\ref{thm:identity} above \texttt{customers}. The next
  step in Figure~\ref{fig:large}.d again uses
  Proposition~\ref{thm:join} to commute \gus{12} and \gusi resulting
  in \gus{121}. Figure~\ref{fig:large}.e shows the final
  transformation that uses the same proposition to get an overall GUS
  method \gus{123} just below the aggregate and on the top of the rest
  of the plan. Theorem~\ref{thm:ev} can now be used to obtain expected
  value and variance of the estimate. Using this and either the normal
  approximation or the Chebychev bounds, we obtain confidence
  intervals for the estimate.

The computed coefficients for the GUS methods involved
  are depicted in Figure~\ref{fig:large}
\end{example}

\begin{figure*}\label{fig:large}
\begin{tabular}{@{\extracolsep{-2em}}cc@{\extracolsep{-0.5em}}ccc} 
\Tree [.\texttt{SUM} [.$\Join$ [.$\Join$ [.$\Join$ [.$B_{0.1}$ [.\texttt{l} ] ] [.$W_{1000}$ [.\texttt{o} ] ] ] [.\texttt{c} ] ] [.$B_{0.5}$ [.\texttt{p} ] ] ] ] & 
\Tree [.\texttt{SUM} [.$\Join$ [.$\Join$ [.$\Join$ [.\gus{1} [.\texttt{l} ] ] [.\gus{2} [.\texttt{o} ] ] ] [.\texttt{c} ] ] [.\gus{3} [.\texttt{p} ] ] ] ] & 
\Tree [.\texttt{SUM} [.$\Join$ [.$\Join$ [.\gus{12} [.$\Join$ [.\texttt{l} ] [.\texttt{o} ] ] ] [.$\gusi$ [.\texttt{c} ] ] ] [.\gus{3} [.\texttt{p} ] ] ] ] & 
\Tree [.\texttt{SUM} [.$\Join$ [.\gus{121} [.$\Join$ [.$\Join$ [.\texttt{l} ] [.\texttt{o} ] ] [.\texttt{c} ] ] ] [.\gus{3} [.\texttt{p} ] ] ] ] & 
\Tree [.\texttt{SUM} [.\gus{123} [.$\Join$ [.$\Join$ [.$\Join$ [.\texttt{l} ] [.\texttt{o} ] ] [.\texttt{c} ] ] [.\texttt{p} ] ] ] ]\\ \\
(a) & (b) & (c) & (d) & (e)
\end{tabular}
\\
\begin{tabular}{ l | p{.8\textwidth}} 
GUS method & Parameters\\ \hline
\gus{1} & $\agus{1} = 0.1, \bgt{1}{\varnothing} = 0.01, \bgt{1}{\texttt{l}} = 0.1$\\
\gus{2} & $\agus{2} = 6.667 \times 10^{-3}, \bgt{2}{\varnothing} = 4.44 \times 10^{-5}, \bgt{2}{\texttt{o}} = 6.667 \times 10^{-3}$\\
\gus{3} & $\agus{3} = 0.5, \bgt{3}{\varnothing} = 0.25, \bgt{3}{\texttt{p}} = 0.5$\\
\gus{12} & $\agus{12} = 6.667 \times 10^{-4}, \bgt{12}{\varnothing} = 4.44 \times 10^{-7}, \bgt{12}{\texttt{o}} = 6.667 \times 10^{-5}, \bgt{12}{\texttt{l}} = 4.44 \times 10^{-6}, \bgt{12}{\texttt{lo}} = 6.667 \times 10^{-4}$\\
\gus{121} & $\agus{121} = 6.667 \times 10^{-4}, \bgt{121}{\varnothing} = 4.44 \times 10^{-7}, \bgt{121}{\texttt{c}} = 4.44 \times 10^{-7}, \bgt{121}{\texttt{o}} = 6.667 \times 10^{-5}, \bgt{121}{\texttt{oc}} = 6.667 \times10^{-5}, \bgt{121}{\texttt{l}} = 4.44 \times 10^{-6}, \bgt{121}{\texttt{lc}} = 4.44 \times 10^{-6}, \bgt{121}{\texttt{lo}} = 6.667 \times 10^{-4}, \bgt{121}{\texttt{loc}} = 6.667 \times 10^{-4}$\\
\gus{123} & $\agus{123} = 3.334 \times 10^{-4}, \bgt{123}{\varnothing} = 1.11 \times 10^{-7}, \bgt{123}{\texttt{p}} = 2.22 \times 10^{-7}, \bgt{123}{\texttt{c}} = 1.11 \times 10^{-7}, \bgt{123}{\texttt{cp}} = 2.22 \times 10^{-7}, \bgt{123}{\texttt{o}} = 1.667 \times 10^{-5}, \bgt{123}{\texttt{op}} = 3.335 \times 10^{-5}, \bgt{123}{\texttt{oc}} = 1.667 \times 10^{-5}, \bgt{123}{\texttt{ocp}} = 3.335 \times10^{-5}, \bgt{123}{\texttt{l}} = 1.11 \times 10^{-6}, \bgt{123}{\texttt{lp}} = 2.22 \times 10^{-6}, \bgt{123}{\texttt{lc}} = 1.11 \times 10^{-6}, \bgt{123}{\texttt{lcp}} = 2.22 \times 10^{-6}, \bgt{123}{\texttt{lo}} = 1.667 \times 10^{-4}, \bgt{123}{\texttt{lop}} = 3.334 \times 10^{-4}, \bgt{123}{\texttt{loc}} = 1.667 \times 10^{-4}, \bgt{123}{\texttt{locp}} = 3.334 \times 10^{-4}$
\end{tabular}
\caption{{\small Transformation of the query plan to allow analysis}}\label{fig:large}
\end{figure*}

\section{Properties of GUS operators}\label{sec:properties}\label{sec:design}
% We need 1.composition - \gus{1} \circ \gus{2} - combines two different lineages (acts like cross product though not exactly that) 2. bag union - union of samples from the same cross product space 3. intersection (behaves like gus compaction) 4. probabilistic mixture

In the previous section we explored the interaction between GUS
operators and relational algebra operators. In this section, we
investigate interactions between GUS operators when applied to the
same data. Intuitively, this will open up avenues for design of
  sampling operators, since it will indicate how to compute GUS
  quasi-operators that correspond to complex sampling schemes.

\begin{proposition}[GUS Union] \label{thm:Un-GUS}
For any expression $R$ and GUS methods \gus{1}, \gus{2},
\begin{gather*}
  \gus{1}(R) \cup \gus{2}(R) \SOA \gus{}(R)\\
  where, \qquad a = a_1 +a_2 -a_1 a_2\\
  b_{T} = 2 a - 1 + (1 - 2 a_1 + b_{1T})(1 - 2 a_2 + b_{2T})
\end{gather*}
\end{proposition}

Union of GUS methods can be very useful when samples are expensive to
acquire, thus there is value in reusing them. If two separate samples
from relation $R$ are available, Proposition~\ref{thm:Un-GUS} provides
a way to combine them.

\begin{proposition}[GUS Compaction] \label{thm:compaction} For any expression $R$, and GUS methods $\gus{1}, \gus{2}$,
\begin{gather*}
  \gus{1}\left(\gus{2}(R)\right) \SOA \gus{}(R),\\
where,\qquad a = \agus{1}\agus{2},\qquad \bt{T} = \bgt{1}{T_1}\bgt{2}{T_2}
\end{gather*}
\end{proposition}
% \begin{proof}
%   The two GUS methods act as successive independent random filters. The
%   probability that a tuple $t$ makes through both GUS methods is the product
%   of the probabilities that it passes through each one. Thus $\agus{} = \agus{1}\agus{2}$.
%   The same reasoning applies to pairs of tuples with overlapping lineage $T$. Thus, 
%  $\bt{T} = \bgt{1}{T}\bgt{2}{T}$. 
% \end{proof}

Compaction can be also viewed as intersection. It allows sampling
methods to be \emph{stacked} on top of each other to obtain smaller
samples. We will make use of this in the next section.

Interestingly, union behaves like $+$ with the null element $\gusz$
(the sampling method that blocks everything), the
compaction/intersection behaves like $*$ with the null element
$\gusi$. Overall, the algebraic structure formed is that of a
\emph{semi-ring}, as stated in the following.

\begin{thm}\label{GUS-algebra}
  The GUS operators over any expression $R$,  form a \emph{semiring}
  structure with respect to the union and compaction operations with
  $\gusz$ and $\gusi$ as the null elements, respectively.
\end{thm} 

The semi-ring structure of GUS methods can be exploited to design
sampling operators from ingredients. 

\begin{proposition}[GUS Composition]\label{thm:composition} 
  For any expressions $R_1$, $R_2$ and \gus{1}, \gus{2}, 
\begin{gather*}
  \gus{1}(R_1) \circ \gus{2}(R_2) \SOA \gus{}(R)\\
  \agus = \agus{1}\agus{2}, \qquad \bt{T} = \bgt{1}{T}\bgt{2}{T}
\end{gather*}
\end{proposition}

GUS concatenation is very useful for design of multi-\\dimensional
sampling operators. We use it here to design a bi-dimensional
Bernoulli.

\begin{example}
  Suppose that we designed a bi-dimensional sampling operator
  $B_{0.2,0.3}(\texttt{l,o}) $ that combines Bernoulli sampling operators
  $B_{0.2}(\texttt{l})$ and $B_{0.3}(\texttt{o})$. Using the above result, the GUS
  operator $\gus{}$ corresponding to the bi-dimensional Bernoulli is
  $\gus{1}(\texttt{l})\circ\gus{2}\texttt{o}$, where $\gus{1}$ is the GUS of $B_{0.2}(\texttt{l})$
  and $\gus{2}$ is the GUS of $B_{0.3}(\texttt{o})$. Working out the
  coefficients using Proposition~\ref{thm:composition} - the process
  is similar to the process in Example~\ref{ex:join} we get:
$\agus{3} = 0.06$, $\bgt{3}{\varnothing} = 0.0036$, $\bgt{3}{\texttt{o}} = 0.012$, $\bgt{3}{\texttt{l}} = 0.018$, $\bgt{3}{\texttt{lo}} = 0.06$
\end{example}

\section{Implementation Considerations}\label{sec:system}\label{implement}

In this section we carefully investigate how the theoretical ideas in
the previous section can be used to add confidence interval
capabilities to existing and future database systems for aggregate
\texttt{'SELECT-FROM-WHERE'} queries.  The main pitfall we are trying
to avoid is the need to \emph{re-design} the query processing engine.
This is the main reason the \emph{online aggregation} type of work
(ripple joins\cite{Haas_ripplejoins}, DBO\cite{dbo}) did not have much industry impact.

As we will see in this section and the further refinement in
Section~\ref{sec:better}, the solution we propose (a) will work with
existing and future sampling methods/operators - the only requirement
is that they are expressible as GUS operators, (b) the analysis is easy
to integrate with existing query optimizers or as a separate tool,
(c) there is no significant restriction on the query plan - the
optimizer is not hindered in the search for a good execution plan, (d)
the analysis needs minimal extra information, and (e) the estimation
process can be confined to a single module that works like a
\emph{black box}.

Our solution, exemplified for Query 1 in Example~\ref{ex:query},  is depicted in
Figure~\ref{fig:architecture}. All the work is performed by the
statistical estimator, denoted as the SBox component, that is
interspersed between the main query plan and the aggregate
computation. The only information the SBox needs is the \emph{lineage}
and the value of the aggregate for each tuple consumed by the
aggregate. Since the SBox needs to perform the aggregation in any
case, the aggregate operator can be omitted; the SBox will provide the
entire result. An enhancement of this solution, that removes the need
to funnel all tuples to the SBox, is explored in
Section~\ref{sec:better}. We use Query 1 in Example~\ref{ex:query} to
make the entire process concrete; this allows us to express the
required computations as SQL statements, which are more familiar and
easier to understand. There is nothing special about Query 1. This
approach works for \emph{any} query supported by our theory with the
appropriate changes to the computation.

%\begin{figure}
%\centering
%\begin{tikzpicture}[scale = 0.5]
%\Tree [.\node(sum){\framebox{SUM}}; ]
%\begin{scope}[shift = {(1.25in, 0.5in)}]
%	\Tree [.Confidence-Interval [.\node(sbox){\framebox{SBox}}; [.\node(join){$\Join$}; [.$B_{0.1}$ [.\texttt{l} ] ] [.$WOR_{1000}$ [.\texttt{o} ] ] ] ] ]
%\end{scope}
%\begin{scope}[shift = {(2.4in, 0.1in)}]
%\Tree [.\node(lqp){\makebox{Logical Query Plan}}; ]
%\end{scope} 
%\begin{scope}[shift = {(2in, -0.15in)}]
%\Tree [.\node(lin){\makebox{Lineage, Agg function}}; ]
%\end{scope}
%\begin{scope}[dashed]
%\draw (sum) -- (join);
%\end{scope}
%\draw[->] (lqp) -- (sbox);
%\end{tikzpicture}
%\caption{{\small Integration of SBox into the query plan}}\label{fig:architecture}
%\end{figure}

There are three tasks that need to be performed by the SBox: use the
query plan and transformations in Section~\ref{sec:analysis} to
compute the coefficients of the top GUS operator, estimate the
coefficients $y_S$ from the samples and perform the final expected
value and variance estimate, and confidence interval computation. We
discuss each part below.

\subsection{Computation of the SOA-equivalent plan}

Given a description of the plan the execution engine will run, and the plan
that uses sampling operators, the theory in Section~\ref{sec:analysis}
is used to compute a SOA-equivalent plan that has a single GUS
operator below the aggregate. The goal of this step is to compute the
coefficients of this GUS operator. No other information is needed for
other parts. This process starts by computing the GUS operators that
correspond to the sampling operators using the technique in
Section~\ref{sec:gus-op} -- this is a simple instantiation process
using Table~\ref{tbl:bwor}. Then, the GUS operators are \emph{pushed up}
the query tree using the transformation rules in
Section~\ref{sec:relGUS}. With careful implementation, this process
need not take more than a few milliseconds even for plans involving 10
relations. At the end of the process, $\gus{}$ is computed. From its
coefficients, using the formula in Theorem~\ref{thm:ev}, the
coefficients $c_S$ are computed with the formula:
\begin{equation*}
  c_S  = \sum_{T \in \mathcal{P}(n)} \left(-1 \right)^{|T| + |S|} b_T
\end{equation*}

\subsection{Lineage information}

As mentioned in Section~\ref{sec:gus-op} the GUS operators require
\emph{lineage} information to express the computation. In the
estimation process, the lineage needs to be made available to the
SBox. In general, adding lineage to databases is a non-trivial issue
\cite{provenanceindb}. Luckily, for GUS operators, we need only a restrictive
version: the lineage of each tuple in a base table is an
ID, the lineage of an intermediate tuple, is the list of IDs
for each base relation tuple that participated. Since we can only
accommodate selection and joins in this work, the lineage of the result
of a join is the concatenation of the lineage of the arguments. In practice, 
all there is needed is to carry IDs of tuples through the query plan and 
make them available, together with the aggregate, to the SBox. For our 
running example, Query 1, this means that the SBox gets the result of the SQL query:
\begin{verbatim}
  CREATE TABLE samples AS
  SELECT l_orderkey*10+l_linenumber as l, 
     o_orderkey as o, l_discount*(1.0-l_tax) as f
  FROM lineitem TABLESAMPLE (10 PERCENT), 
     orders TABLESAMPLE(1000 ROWS)
  WHERE l_orderkey  = o_orderkey AND 
     l_extendedprice > 100.0;
\end{verbatim}
The IDs of tuples need to be unique for each tuple in a base
relation. If the database engine maintains row ids internally, they
can be used. If not, as is the case here, the attributes forming
the primary key can be used to compute an ID, either through some
computation -- this is the case for \texttt{lineitem} above, or
through the application of a hash function with a large domain. As
required by the theory, the only operation the SBox is allowed to
perform is comparison of IDs, thus any one-to-one mapping suffices.

In some systems, the extra lineage information might add significant
overhead. Section~\ref{sec:better} deals with the issue and allows
further improvement.

\subsection{Estimating terms $y_S$}

The computation of the variance of the sampling estimator in
Theorem~\ref{thm:ev} uses the coefficients $y_S$ defined as:
\begin{equation*}
  y_S = \sum_{t_i \in R_i| i \in S} \left(\sum_{t_j \in R_j| j \in S^C}f({t_i, t_j})\right)^2.
\end{equation*}
 The terms $y_S$ essentially requires a \emph{group by lineage}
 followed by a specific computation. This is better understood through
 an example -- Query 1 -- and equivalent expressions in SQL:
\begin{verbatim}
  CREATE TABLE unagg AS
  SELECT  l_orderkey*10+l_linenumber as l, 
     o_orderkey as o, l_discount*(1.0-l_tax) as f
  FROM lineitem TABLESAMPLE (10 PERCENT), 
     orders TABLESAMPLE(1000 ROWS)
  WHERE l_orderkey  = o_orderkey AND 
     l_extendedprice > 100.0;

  SELECT sum(f)^2 as y_empy FROM unagg;

  SELECT sum(F*F) as y_l FROM ( SELECT sum(f) as F 
         FROM unagg GROUP BY l);

  SELECT sum(F*F) as y_o FROM ( SELECT sum(f) as F 
         FROM unagg GROUP BY o);

  SELECT sum(f*f) as y_lo FROM unagg; 
\end{verbatim}

The computation of the $y_S$ terms using the above code is harder than
the evaluation of the exact query, thus resulting in an impractical
solution. We can use the sample provided to the SBox to \emph{estimate
  these terms} by essentially replacing \texttt{unagg} above with
\texttt{samples}. These estimates, $Y_S$ can be used to obtain
unbiased estimates $\hat{Y}_S$ of terms $y_S$ using the
formula\cite{DBLP:journals/tods/JermaineAPD08}
\begin{equation*}
 \hat{Y}_S = \frac{1}{c_{S,\emptyset}} \left(Y_S - \sum_{T \subset S^C, T\neq\emptyset} c_{S,T} \hat{Y}_{S\cup T} \right)
\end{equation*}
where 
\begin{equation*}
  c_{S,T} =  \sum_{U \subset T} \left(-1 \right)^{|U| + |S|} b_{S\cup U}.
\end{equation*}
Note that the major effort is in evaluating $Y_S$ terms
over the sample - the rest of the computation only depends on the
number of the relations.

\subsection{Confidence interval computation}

Once the $\hat{Y}_S$ estimates of $y_S$ are computed, the variance
formula in Theorem~\ref{thm:ev} can be used. In particular, the
estimate of variance of the sampling estimate is:
\begin{equation*}
  \hat{\sigma^2} = \sum_{S \subset \{1:n\}}\frac{c_S}{a^2}\hat{Y}_S - \hat{Y}_{\phi}.
\end{equation*}
To produce actual confidence intervals, we can use one of the
following techniques:

{\bf Optimistic confidence intervals} In most circumstances the
distribution of the sampling estimate is very close to \emph{normal
 distribution}. The techniques in this paper allow the computation of
estimates of the expected value $\hat{\mu}$ and variance
$\hat{\sigma^2}$. The concrete formula for a $95$ percent confidence interval is:
%\begin{equation*}
$$[\hat{\mu} - 1.96\hat{\sigma}, \hat{\mu} + 1.96\hat{\sigma}]$$
%\end{equation*}

{\bf Pessimistic Chebychev confidence intervals} If the 
normality of the distribution of sampling estimate is doubtful, the
Chebychev bound can be used to provide 95 percent confidence interval using:
%\begin{equation*}
$$[\hat{\mu} - 4.47\hat{\sigma}, \hat{\mu} + 4.47\hat{\sigma}]$$
%\end{equation*}

The Chebychev confidence intervals are correct for any distribution,
at the expense of a factor of 2 in width.

\section{Efficient Variance Estimator}\label{sec:better}

As we explained in Section~\ref{sec:system}, the estimator of the true
result, the expected value of the sampling estimator, does not require
any lineage information. It is simply a scaled version of the result
of the query containing sampling. When it comes to the variance
estimate, there are two main concerns when the number of result tuples
before aggregation is large: (a) the number of terms to be evaluated
is $2^n$ where $n$ is the number of base relations, and (b) ea6ch term
consists of a \texttt{GROUP BY} query that is possibly expensive. In
this section we address these problems using the extension of the
base theory in Section~\ref{sec:design}. 

We start by making an observation about the computation of the
variance of the sampling estimator: it depends, in orthogonal ways, on
properties of the data through terms $y_S$ and on properties of the
sampling through $c_S$. The base theory does not require any
particular way to compute/estimate terms $y_S$. \emph{Using the
  available sample for estimating $y_S$ terms is only one of the
  possibilities.} While many ways to estimate terms $y_S$ can be
explored, a particularly interesting one in this context is to use
\emph{another sampling method} for the purpose. More specifically, we
could use a sample of the available sample for estimation of the terms
$y_S$ and the full sample for the estimation of the true value. This
process is depicted in Figure~\ref{fig:better}. 

%\begin{figure}
%\centering
%\begin{tikzpicture}[scale = 0.5]
%\tikzstyle{every node}=[font=\small]
%\Tree [.\node(sum){\framebox{SUM}}; [.$\Join$ [.$B_{0.1}$ [.\texttt{l} ] ] [.$WOR_{1000}$ \texttt{o} ] ] ]
%\begin{scope}[shift = {(1.25in, 0.416in)}]
%	\Tree [.Confidence-Interval [.\node(sbox){\framebox{SBox}}; [.$\Join$ [.$B_{0.01}$ [.\texttt{l} ] ] [.$WOR_{100}$ [.\texttt{o} ] ] ] ] ]
%\end{scope}
%\begin{scope}[shift = {(2.25in, 0.01in)}]
%\Tree [.\node(lqp){\makebox{Logical Query Plan}}; ]
%\end{scope} 
%\begin{scope}[shift = {(2in, -0.25in)}]
%\Tree [.\node(lin){\makebox{Lineage, Agg function}}; ]
%\end{scope}
%\draw[->] (sum) -- (sbox);
%\draw[->] (lqp) -- (sbox);
%\end{tikzpicture}
%\caption{{\small Alternative design with two sampling methods. The smaller is
%  used to estimate the terms $y_S$, the larger for the main estimate.}}\label{fig:better}
%\end{figure}

To understand what benefits we can get from this idea, we observe that
we do not need very precise estimates of the terms $y_S$. Should we
make a mistake, it will only affect the confidence interval by a small
constant factor but will still allow the shrinking of the confidence
interval with the increase of the sample. Based on the experience in
DBO and TurboDBO, using 10000 result tuples for the estimation of
$y_S$ terms suffices. This means that the $2^n$ $y_S$ terms are
evaluated, as explained in Section~\ref{sec:system} only on datasets
of size at around 10000. Moreover, only for these 10000 samples the
system needs to provide lineage information; samples used for
evaluation of the expected value need no lineage.

There are two alternatives when it comes to reducing the number of
samples used for estimation of terms $y_S$: select a more restrictive
sampling method, depicted in Figure~\ref{fig:better}, or further
sample from the provided sample. The later approach can be applied when
needed in case in which the size of the sample is overwhelming for the
computation of terms $y_S$. Specifically, we can use a
multi-dimensional Bernoulli GUS on top of the existing query plan for 
result tuples. This can be obtained by applying Proposition~\ref{thm:composition} 
until desired dimension is reached. The extra results in Section~\ref{sec:properties} 
together with the core results in Section~\ref{sec:analysis} provide the means to analyze
this modified sampling process. Example~\ref{ex:bern} and the
accompanying Figure~\ref{fig:bern} provide such analysis for Query 1
and exemplifies how the extra Bernoulli sampling can be dealt
with. 
\begin{example}\label{ex:bern}
This example shows how the query plan for Query 1 can be sampled
further to efficiently obtain $y_S$ terms. Figure~\ref{fig:bern}.a shows the original query plan.
Figure~\ref{fig:bern}.b shows the sampling in terms of a GUS quasi-operator. Figure~\ref{fig:bern}.c 
shows the placement of a bi-dimensional Bernoulli sampling method. Figures~\ref{fig:bern}.d, ~\ref{fig:bern}.e, ~\ref{fig:bern}.f make use of propositions in Section~\ref{sec:analysis} to obtain 
a SOA-equivalent plan, suitable for analysis.
\end{example}

\begin{figure*}[tbhp]
\centering
\begin{tabular}{@{\extracolsep{0.001em}}cc|c@{\extracolsep{-0.5em}}ccc}
\Tree [.\texttt{SUM} [.$\Join$ [.$B_{0.1}$ [.\texttt{litem} ] ] [.$WOR_{1000}$ [. \texttt{ord} ] ] ] ] &
\Tree [.\texttt{SUM} [.\gus{BW} [.$\Join$ [.\texttt{litem} ]
[.\texttt{ord} ] ] ] ] &
\Tree [.\texttt{SUM} [.$B_{0.2,0.3}$ [.$\Join$ [.$B_{0.1}$
[.\texttt{l} ] ] [.$WOR_{1000}$ [.\texttt{o} ] ] ] ] ] &
\Tree [.\texttt{SUM} [.\gus{3} [.$\Join$ [.\gus{1}
[.\texttt{l} ] ] [.\gus{2} [.\texttt{o} ] ] ] ] ] & 
\Tree [.\texttt{SUM} [.\gus{3} [.\gus{12} [.$\Join$ [.\texttt{l} ] [.\texttt{o} ] ] ] ] ]&
\Tree [.\texttt{SUM} [.\gus{123} [.$\Join$ [.\texttt{l} ]
[.\texttt{o} ] ] ] ]\\
 & & & & & \\
(a) & (b) & (c) & (d) & (e) & (f)
\end{tabular}

\vspace{\baselineskip}
\begin{tabular}{ l | p{.8\textwidth}}
	GUS method & Parameters\\ \hline
	\gus{1} & $\agus{1} = 0.1, \bgt{1}{\varnothing} = 0.01, \bgt{1}{\texttt{l}} = 0.1$\\
	\gus{2} & $\agus{2} = 6.667 \times 10^{-3}, \bgt{2}{\varnothing} = 4.44 \times 10^{-5}, \bgt{2}{\texttt{o}} = 6.667 \times 10^{-3}$\\
	\gus{3} & $\agus{3} = 0.06, \bgt{3}{\varnothing} = 0.0036, \bgt{3}{\texttt{o}} = 0.012, \bgt{3}{\texttt{l}} = 0.018, \bgt{3}{\texttt{lo}} = 0.06$\\ 
	\gus{12} & $\agus{12} = 6.667 \times 10^{-4}, \bgt{12}{\varnothing} = 4.44 \times 10^{-7}, \bgt{12}{\texttt{o}} = 6.667 \times 10^{-5}, \bgt{12}{\texttt{l}} = 4.44 \times 10^{-6}, \bgt{12}{\texttt{lo}} = 6.667 \times 10^{-4}$\\
	\gus{123} & $\agus{123} = 4 \times 10^{-5}, \bgt{123}{\varnothing} = 1.598 \times 10^{-9}, \bgt{123}{\texttt{o}} = 8 \times 10^{-7}, \bgt{123}{\texttt{l}} = 7.992 \times 10^{-8}, \bgt{123}{\texttt{lo}} = 4 \times 10^{-5}$ 
\end{tabular}
\caption{Transformation of the query plan to allow analysis} \label{fig:bern}
\end{figure*}

For the estimation process to be correct, the Bernoulli sampling out
of the sample computed by the query plan needs some care in
implementation. The main issue is the fact that it has to be a GUS
method - if it decides to eliminate a tuple from a base relation, it
has to do so in all result tuples in which it appears. This can be
easily achieved efficiently and with little space using pseudo-random
functions that combine seeds and lineage to provide a [0,1]
number. The pseudo-randomness ensures that the value of the function
will return the same value for the same tuple, thus providing the same
decision. The only memory required to run such a sub-sampling algorithm
is minimal: one seed per base relation. The process is also very
efficient since it only requires evaluation of simple functions.

%\section{Discussion}\label{sec:applic} 
\section{Conclusions and Future Work}\label{sec:conclusion}

While technically challenging to create, the theory in this paper is in essence easy to use. Sampling is treated as a quasi-operator. In order to incorporate sampling based approximations, such operators are introduced in the query plans and the mechanisms described in Section~\ref{sec:analysis} are used to analyze the estimators. We have already seen an example of use of the theory: the sub-sampling technique in Section~\ref{sec:implementation}. With very little effort (introducing a final Bernoulli sampling quasy operator), we dealt with a seemingly hard problem: how to use a subsample to predict the behavior of the main sample. The straightforwardness of this process encourages us to suggest that the theory presented here will allow significant progress in a number of hard to solve problems explored in the approximate query processing literature. We briefly mention such potential in the remaining of this section.

\vspace{5 pt}
\noindent
{\bf Database as a sample}. By viewing the database itself as a sample, \emph{robustness analysis} is possible. In particular, if  we assume that 1\% of the tuples are mistakenly lost and we wish to predict the impact on the query results we can view the database as a 99\% Bernoulli sample. A large variance will indicate that the query results are sensitive to such perturbations and thus not robust.

\vspace{5 pt}
\noindent
{\bf Choosing sampling parameters}. By using the unbiased $y_S$ estimates from a single sampling instance, the theory allows for plugging in co-efficients for different sampling strategies to predict the respective variances. This can give the user insight on comparing diffrent sampling strategies and parameters to suit his/her needs.

\vspace{5 pt}
\noindent
{\bf Estimating the size of intermediate relations}. Query execution engines maintain a sample of the data and evaluate aggregates on it to predict the size of the intermediate relations. Our theory allows for the evaluation of the precision of these, thereby preventing the selection of \emph{inferior} plans.  

\vspace{5 pt}
\noindent
{\bf Data Streaming and Load Shedding}. An interesting problem in load shedding is determing a sampling rate so that the the system can keep up with fast-rate incoming data while minimizing the error\cite{DBLP:series/ads/BabcockDM07}. While such analysis was done for single relations, our theory provides for similar analyis with multiple relations.

%This is the problem that we used as the motivation for our experiments in Section~\ref{sec:exp}. When the database itself is a sample of a larger data set who's behavior we want to predict. An immediate practical application is a form of \emph{robustness analysis} in which we assume that 1\% of the tuples are mistakenly lost and want to predict the impact on the query answers. Should a large variation in the answer result, the answer is proved to be sensitive to such small perturbations thus not robust. Applying our theory to this problem is a simple matter of 

\section{Discussion future directions}\label{sec:disc}

In this section we briefly mention possible theoretical developments
that use the current work as a starting point. 

\paragraph*{Random sets} Proposition~\ref{thm:SOA-P} establishes a
connection between SOA-equivalence and an equivalence relation on
random sets: two random sets are equivalent if they agree on
probabilities of inclusions of tuples and pairs of tuples. This
equivalence is a relaxed version of the equivalence relationship that
would require the \emph{same distribution}. When SUM-like aggregates
are computed over random sets, this relationship between random sets
is the same as SOA-equivalence, as stated in
Proposition~\ref{thm:SOA-P}. An interesting question to explore is
what other \emph{properties} are preserved by this equivalence. This
might prove to be the key to tackling non-aggregate queries in the
presence of randomization.

\paragraph*{Extending randomized filtering.} In Section~\ref{sec:gus}
we saw that GUS is required to be a \emph{filter}. This means that
sampling methods that produce duplicates, like sampling with
replacement, are not accommodated. More importantly, the filter
behavior required the removal of duplicates in
Proposition~\ref{thm:Un-GUS}, which is potentially costly. We believe
the current theory can be extended naturally to allow and capture such
duplicates and still maintain the current elegance.

\paragraph*{Average and non-linear combinations of SUM-like
  aggregates}. Our theory works only for SUM-like aggregates. For
non-linear combinations of such aggregates such as \texttt{AVERAGE},
it is not possible to compute exactly the moment but there is a good
chance that good quality approximations can be provided, using for
example the \emph{delta method}\cite{stat-book}. An interesting
question is whether the structure in this paper can be extended for
such approximation, especially since the Taylor expansions used in the
delta method have the same linear structure. 

\paragraph*{Dealing with Self-Joins:} Proposition~\ref{thm:join}
requires no overlap in the lineage of the two arguments of the
join. This effectively bans self-joins. Self-joins introduce
difficulties for probabilistic analysis - this happens for
probabilistic databases- since the probabilistic event
specifying the presence of a tuple is used twice in the
estimation. This creates extra dependencies that are not fully
captured by inclusion probabilities of tuples and pairs of tuples as
is the case for GUS. An interesting question is whether more but a
finite extra inclusion probabilities are enough to deal with the
problem. We conjecture that inclusion probabilities of combinations of
4 tuples is sufficient.

\paragraph*{Dealing with \texttt{DISTINCT}} The GUS family of sampling
methods is not general enough to commute with distinct -- counter
examples can be readily build. The main problem is the fact that the
DISTINCT needs more information than the interaction between two
tuples, even for the computation of expected value. We believe that
there is a deep connection between this problem and the \emph{safe
  plans} in probabilistic databases. An interesting future development 
would be identifying a more general sampling class
then GUS and establishing a connection with the safe plans. 

%\paragraph*{MIN,MAX aggregates} ...

%\paragraph*{Dealing with Group-By} Providing confidence intervals for
%each of the groups in a \texttt{GROUP BY} clause can readily use the
%theory in this paper. The estimates for the groups will be correlated
%bu the covariance can be computed much like the variance can
%\cite{DBO-journal}. Using the \emph{simultaneous inference}\cite{SI},
%comprehensive characterization of group by queries can be provided. The
%only difficulty is the question of the number of groups, similar to
%the \texttt{DISTINCT} problem. 

\appendix

\section{SOA-equivalence: proof}
{\it Proof} Suppose $\mathcal{E}(R) \SOA \mathcal{F}(R)$. For every $t \in R$, define 
the function $f_t$ as $f_t (s) = 1_{\{s = t\}}$. Hence, $\mathcal{A}_{f_t} (S) = 
1_{\{t \in S\}}$. It follows that 
\begin{eqnarray*}
P(t \in \mathcal{E} (R)) 
&=& E \left[ \mathcal{A}_{f_t}( \mathcal{E}(R) ) \right]\\
&=& E \left[ \mathcal{A}_{f_t} ( \mathcal{F}(R) ) \right]\\
&=& P(t \in \mathcal{F} (R)). \label{eq1}
\end{eqnarray*}

\noindent
Now, for every $t,t' in R$, define the function $f_{t,t'}$ as $f_{t,t'} (s) = 
1_{\{s = t\}} + 1_{\{s = t'\}}$. It follows that 
\begin{eqnarray*}
& & E \left[ \mathcal{A}_{f_t} (\mathcal{E}(R))^2 \right]\\
&=& E \left[ \left( 1_{\{t \in \mathcal{E}(R)\}} + 1_{\{t' \in \mathcal{E}(R)\}} 
\right)^2 \right]\\
&=& E \left[ 1_{\{t \in \mathcal{E} (R)\}} + 1_{\{t' \in \mathcal{E} (R)\}} + 
2 1_{\{t,t' \in \mathcal{E} (R)\}} \right]\\
&=& P(t \in \mathcal{E} (R)) + P(t' \in \mathcal{E} (R)) + 
2 P(t,t' \in \mathcal{E} (R)). 
\end{eqnarray*}

Similarly, 
\begin{eqnarray*}
& & E \left[ \mathcal{A}_{f_t} (\mathcal{F} (R))^2 \right]\\
&=& P(t \in \mathcal{F} (R)) + P(t' \in \mathcal{F} (R)) + 
2 P(t,t' \in \mathcal{F} (R)). 
\end{eqnarray*}

\noindent
Note that 
$$
E \left[ \mathcal{A}_{f_t} (\mathcal{E} (R))^2 \right] = E \left[ 
\mathcal{A}_{f_t} (\mathcal{F} (R))^2 \right]. 
$$

\noindent
It follows by (\ref{eq1}) that 
$$
P(t,t' \in \mathcal{E} (R)) = P(t,t' \in \mathcal{F} (R)). 
$$

Hence, one direction of the equivalence is proved. Let us now assume 
that 
$$
P(t \in \mathcal{E} (R)) = P(t \in \mathcal{F} (R)) \; \forall t \in R, 
$$

\noindent
and 
$$
P(t,t' \in \mathcal{E} (R)) = P(t,t' \in \mathcal{F} (R)) \; \forall t,t' \in R. 
$$

\noindent
The SOA-equivalence of $\mathcal{E} (R)$ and $\mathcal{F} (R)$ immediately follows 
by noting that for an arbitray function $f$ on $R$, and an arbitrary 
(possibly randomized) expression $S(R)$ 
$$
E \left[ \mathcal{A}_f (S(R)) \right] = \sum_{t \in R} P(t \in S(R) f(t), 
$$

\noindent
and 
$$
E \left[ \mathcal{A}_f (S(R))^2 \right] = \sum_{t,t' \in R} P(t,t' \in S(R) f(t) f(t'). 
$$

\hfill$\Box$ 

\bibliographystyle{abbrv}
\bibliography{supriya}  

\end{document}